\documentclass[10pt]{article}%{elsart}%amsmath,amssymb]%{revtex4}
\usepackage{amssymb,amsmath}
\usepackage{amsthm}
\usepackage{mathrsfs}  
\usepackage{graphicx}
\usepackage{dcolumn}
\usepackage{bm}
\usepackage{fullpage}
\usepackage{color}
\usepackage[all]{xy}
\begin{document}
%%%%%%%%%%%%%%%%
\newtheorem{Def}{Definition}[section]
\newtheorem{Thm}{Theorem}[section]
\newtheorem{Proposition}{Proposition}[section] 
\newtheorem{Lemma}{Lemma}[section]
\theoremstyle{definition}
\newtheorem*{Proof}{Proof}%[section]
\newtheorem{Example}{Example}%[section]  
\newtheorem{Postulate}{Postulate}[section]
\newtheorem{Corollary}{Corollary}[section]
\newtheorem{Remark}{Remark}[section]
\theoremstyle{remark} 
\newcommand{\beq}{\begin{equation}}
\newcommand{\beqa}{\begin{eqnarray}}
\newcommand{\eeq}{\end{equation}}
\newcommand{\eeqa}{\end{eqnarray}}
\newcommand{\non}{\nonumber}
\newcommand{\fr}[1]{(\ref{#1})}
\newcommand{\cc}{\mbox{c.c.}}
\newcommand{\nr}{\mbox{n.r.}}
\newcommand{\eq}{\mathrm{eq}}
\newcommand{\B}{\mathrm{B}}
\newcommand{\Ising}{\mathrm{Ising}}
\newcommand{\heatbath}{\mathrm{heat\, bath}}
\newcommand{\ext}{\mathrm{ext}}
\newcommand{\atan}{\mathrm{Tan}{}^{-1}}
\newcommand{\atanh}{\mathrm{Tanh}{}^{-1}}
\newcommand{\acosh}{\mathrm{Cosh}{}^{-1}}
\newcommand{\bb}{\mbox{\boldmath {$b$}}}
\newcommand{\bbe}{\mbox{\boldmath {$e$}}}
\newcommand{\bt}{\mbox{\boldmath {$t$}}}
\newcommand{\bn}{\mbox{\boldmath {$n$}}}
\newcommand{\br}{\mbox{\boldmath {$r$}}}
\newcommand{\bC}{\mbox{\boldmath {$C$}}}
\newcommand{\bH}{\mbox{\boldmath {$H$}}}
\newcommand{\bp}{\mbox{\boldmath {$p$}}}
\newcommand{\bx}{\mbox{\boldmath {$x$}}}
\newcommand{\bF}{\mbox{\boldmath {$F$}}}
\newcommand{\bT}{\mbox{\boldmath {$T$}}}
\newcommand{\bomega}{\mbox{\boldmath {$\omega$}}}
\newcommand{\ve}{{\varepsilon}}
\newcommand{\e}{\mathrm{e}}
\newcommand{\F}{\mathrm{F}}
\newcommand{\Loc}{\mathrm{Loc}}
\newcommand{\Ree}{\mathrm{Re}}
\newcommand{\Imm}{\mathrm{Im}}
\newcommand{\hF}{\widehat F}
\newcommand{\hL}{\widehat L}
\newcommand{\tA}{\widetilde A}
\newcommand{\tB}{\widetilde B}
\newcommand{\tC}{\widetilde C}
\newcommand{\tL}{\widetilde L}
\newcommand{\tK}{\widetilde K}
\newcommand{\tX}{\widetilde X}
\newcommand{\tY}{\widetilde Y}
\newcommand{\tU}{\widetilde U}
\newcommand{\tZ}{\widetilde Z}
\newcommand{\talpha}{\widetilde \alpha}
\newcommand{\te}{\widetilde e}
\newcommand{\tv}{\widetilde v}
\newcommand{\ts}{\widetilde s}
\newcommand{\tx}{\widetilde x}
\newcommand{\ty}{\widetilde y}
\newcommand{\ud}{\underline{\delta}}
\newcommand{\uD}{\underline{\Delta}}
\newcommand{\chN}{\check{N}}
\newcommand{\cA}{{\cal A}}
\newcommand{\cB}{{\cal B}}
\newcommand{\cC}{{\cal C}}
\newcommand{\cD}{{\cal D}}
\newcommand{\cF}{{\cal F}}
\newcommand{\cH}{{\cal H}}
\newcommand{\cI}{{\cal I}}
\newcommand{\cL}{{\cal L}}
\newcommand{\cM}{{\cal M}}
\newcommand{\cN}{{\cal N}}
\newcommand{\cO}{{\cal O}}
\newcommand{\cP}{{\cal P}}
\newcommand{\cR}{{\cal R}}
\newcommand{\cS}{{\cal S}}
\newcommand{\cY}{{\cal Y}}
\newcommand{\cZ}{{\cal Z}}
\newcommand{\cU}{{\cal U}}
\newcommand{\cV}{{\cal V}}
\newcommand{\dr}{\mathrm{d}}
\newcommand{\sech}{\mathrm{sech}}
\newcommand{\Exp}{\mathrm{Exp}}
 \newcommand{\GamLamM}[1]{{\Gamma\Lambda^{{#1}}\,\cal{M}}}
 \newcommand{\GTM}{{\Gamma T\cal{M}}}
\newcommand{\inp}[2]{\left\langle\,  #1\, , \, #2\, \right\rangle}
\newcommand{\equp}[1]{\overset{\mathrm{#1}}{=}}
\newcommand{\wt}[1]{\widetilde{#1}}
\newcommand{\wh}[1]{\widehat{#1}}
\newcommand{\ch}[1]{\check{#1}}
\newcommand{\ii}{\imath}
\newcommand{\ic}{\iota}
\newcommand{\scrH}{\mathscr{H}}
\newcommand{\mi}{\,\mathrm{i}\,}
\newcommand{\mr}{\,\mathrm{r}\,}
\newcommand{\mbbC}{\mathbb{C}}
\newcommand{\mbbE}{\mathbb{E}}
\newcommand{\mbbR}{\mathbb{R}}
\newcommand{\mbbZ}{\mathbb{Z}}
\newcommand{\ol}[1]{\overline{#1}}
\newcommand{\rmC}{\mathrm{C}}
\newcommand{\rmH}{\mathrm{H}}
\newcommand{\Id}{\mathrm{Id}} 
\newcommand{\avg}[1]{\left\langle\,{#1}\, \right\rangle}
\newcommand{\Leg}{\mathbb{L}}
\newcommand{\avgg}[1]{\left\langle\langle\,{#1}\, \rangle\right\rangle}
%%%%%%%%%%%%%%%%%%%%%%%%%%%%%%%%%%%%%%%%%%%%%%%%%%%%%%
\title{
Hessian-information geometric 
formulation of a class of deterministic neural network models 
}
%%%%%%%%%%%%%%%%%%%%%%%%%%%%%%%%%%%%%%%%%%%%%%%%%%%%%%
\author{  Shin-itiro GOTO,\\ 
The Institute of Statistical Mathematics,\\ 
10-3 Midori-cho, Tachikawa, Tokyo 190-8562, Japan
} 
%Email: \{sgoto,kazuyuki\}@cslab.kecl.ntt.co.jp
\date{\today}
\maketitle
%%%%%%%%%%%%%%%%%
\begin{abstract}%
%%%%%%%%%%%%%%%%%
In this paper a class of dynamical systems 
describing deterministic neural network models  
are formulated from a viewpoint of differential geometry.
This class includes the 
Hopfield model 
and gradient systems, and 
is such that the so-called activation functions induce 
information and Hessian geometries. 
In this formulation, 
it is shown that the phase space compressibility of 
a dynamical system belonging to this class 
is written in terms of the 
Laplace operator defined on Hessian manifolds, where 
phase space compressibility is associated with 
a volume-form of a manifold, 
and expresses how such a volume-form is compressed along the 
vector field of a dynamical system.   
Since the sigmoid function, as an activation function, 
plays a role in the study of 
neural network models, such compressibility 
is explicitly calculated for this case.
Throughout this paper, the so-called dual coordinates known in 
information geometry are explicitly used. 
%%%%%%%%%%%%%%% 
\end{abstract}%
%%%%%%%%%%%%%%%
%%%%%%%%%%%%%%%%%%%%%%%%%%
\section{Introduction}
%%%%%%%%%%%%%%%%%%%%%%%%%
Neural network models play fundamental roles in brain science to  
clarify functions of human brains \cite{Geszti1990,Amit1989}. 
They can be viewed as 
a class of dynamical systems that could be either probabilistic or 
deterministic,  
and can mathematically be studied. 
Considerable activity is being devoted to the comprehension of 
dynamics of genuine neural network models and their variants.  
Although models involving probability seem to be natural, 
analysis of deterministic models is expected to provide simple perspective
due to simplicity of the models. 
Aside from its purely academic interest, its resolution 
has implications in mathematical engineering. 
In particular neural network models are intensively employed in
machine learning systems \cite{Bishop2006}.
  
Various approaches exist to analyze dynamical systems.  
One of such is to employ differential geometry \cite{Arnold, Silva2008}, 
since 
various geometric notions can be introduced 
and can systematically be applied   
to dynamical systems. 
Among them, one approach is to apply information geometry,  
where information geometry is a geometrization of statistics for 
parametric  models with which 
cumulant generating functions are associated \cite{AN,Ay2017}.
One of essential roles of information geometry is to bridge  
convex analysis and Riemannian geometry, where 
tools in convex analysis includes the Legendre transform. 
Hessian geometry is such a differential geometry involving convex functions 
and it is not necessary to involve probability distribution functions 
\cite{Shima2007}.
Thus, if systems involving convex functions 
are nothing to do with probability, then 
it is expected that Hessian geometry is a key to develop geometrization for 
those systems. Such examples of geometrization  
include Hamiltonian dynamical systems 
and electric circuit models \cite{Goto2016,GW2018}.

For deterministic neural network models involving convex functions, 
it should be expected that convex analysis and its related geometries, 
Hessian and information geometries, play roles.   
To see this explicitly, a class of dynamical systems should be focused.
A candidate of such a class   
is the one proposed by Cohen and Grossberg \cite{Cohen1983}. 
This class includes the 
Hopfield model \cite{Hopfield1984}, and   
this class has been 
intensively studied in the literature  since 
it has Lyapunov functions \cite{Grossberg1988}. 
On the other hand one of the keys in the study of neural network models 
is to choose the so-called activation function. 
However the use of 
convex analysis and its differential geometry based on 
activation functions of this model 
have not been focused.
Thus a formulation based on 
Hessian and information geometries in conformity with a given convex 
activation function is expected to prove fruitful.

In this paper a class of dynamical systems 
including the 
 Hopfield 
model for describing a neural network model
is formulated in terms of Hessian and information geometries. 
In this formulation a coordinate free description of the class of 
dynamical systems is given, from which the quantity called 
phase space compressibility is shown to be written as the negative of the 
Laplace operator acting on 
a Lyapunov function. 
This quantity expresses how much a volume-form is compressed along the  
vector field associated with a dynamical system, and 
is a measure for how fast flow in phase space of the 
dynamical system converges to an attracting set.

This paper is organized as follows. 
In Section\,\ref{section-model}, 
a class of dynamical systems, which we call the generalized 
Hopfield model, 
is introduced.  After a coordinate free description of the class 
is given, phase space compressibility is calculated. 
In Section\,\ref{section-examples},
explicit expressions of phase space compressibility for 
some examples are given. 
Finally, Section\,\ref{section-conclusion} summarizes 
this paper and discusses some future works.

In this paper mathematical objects are assumed to be smooth and real.

%%%%%%%%%%%%%%%%%%%%%%%%
\section{Generalized Hopfield model}
\label{section-model}
%%%%%%%%%%%%%%%%%%%%%%%

In this section a generalized 
 Hopfield 
model is introduced, and then 
its geometric formulation is given. 

Let $\cM$ be an $n$-dimensional manifold, 
$U$ a set of local coordinates on $\cM$ with 
$U=\{U^{\,1},\ldots,U^{\,n}\}$, and $\Psi$ a strictly convex function on $\cM$.
Strictly convexity of $\Psi$ is written as  
$(\partial^{2}\psi/\partial \,U^{\,a}\partial U^{\,b})\succ 0$  
in some convex domain of $\cM$, where 
$(A_{\,ab})\succ 0$ denotes that the matrix $(A_{\,ab})$ is positive definite.
Introduce $V$ with $V=\{V_{\,1},\ldots,V_{\,n}\}$ such that  
\beq
V_{\,a}
=\frac{\partial\Psi}{\partial U^{\,a}}.
\label{V-U-psi}
\eeq

Consider the dynamical system of the form 
\beq
\frac{\dr}{\dr t}U^{\,a}
=F^{\,a}(V),\qquad 
F^{\,a}(V)
=-\,\frac{\partial\cH}{\partial V_{\,a}},
\label{generalized-Cohen-Grossberg-model}
\eeq
where $t\in I\subseteq\mbbR$ plays a role of time, 
and 
$\cH$ is a Lyapunov function:
\beq
\frac{\dr\cH}{\dr t}
=-\,\sum_{a}\sum_{b}\frac{\dr U^{\,a}}{\dr t}
\frac{\partial^{\,2}\,\Psi}{\partial U^{\,a}\partial U^{\,b}}
\frac{\dr U^{\,b}}{\dr t}
<0.
\label{generalized-Cohen-Grossberg-model-H-differential}
\eeq
This system is termed in this paper as follows.
%%%%%%%%%%%%%%%
\begin{Def}
%%%%%%%%%%%%
The dynamical system \fr{generalized-Cohen-Grossberg-model} together with 
\fr{V-U-psi} and \fr{generalized-Cohen-Grossberg-model-H-differential} is referred to as the 
generalized 
Hopfield 
model.
%%%%%%%%%%%%%
\end{Def}
%%%%%%%%%%%%
The following are some examples.
%%%%%%%%%%%%%%%%%
\begin{Example}
\label{example-simplest}
%%%%%%%%%%%%%%%
Choosing $\cH$ to be  
$\cH(V)=\sum_{a}\sum_{b}\delta^{\,ab}\,V_{\,a}V_{\,b}/2$ with $\delta^{\,ab}$ being 
the Kronecker delta giving unity for $a=b$ otherwise vanishes,  
one has a gradient system:
$$
\frac{\dr}{\dr t}U^{\,a}
=-\sum_b\delta^{\,ab}\,V_{\,b}
=-\,\sum_{b}\delta^{\,ab}\,\frac{\partial\,\Psi}{\partial U^{\,b}}.
$$
%%%%%%%%%%%%%%%%%
\end{Example}
%%%%%%%%%%%%%%%%%
\begin{Example}
\label{example-Hopfield}
%%%%%%%%%%%%%%%
A neural network model having a Lyapunov function 
found in the literature is of the form of the set of 
differential equations 
$$
\frac{\dr}{\dr t}U^{\,a}
=\sum_{b}J^{\,ab}\,V_{\,b}
-\frac{1}{R_{\,a}}U^{\,a}+I_{\,\ext}^{\,a},\qquad 
V_{\,a}
=\Upsilon(U^{\,a}),\quad 
\Upsilon(U^{\,a})
:=\frac{\partial\, \Psi}{\partial U^{\,a}}
=\frac{\dr\, \psi}{\dr U^{\,a}},\quad
\Psi(U)
=\sum_{a}\psi(U^{\,a}),
$$
where $\{R_{\,a}\}$ and $\{I_{\,a}^{\,\ext}\}$  
are sets of constants,  $\{J^{\,ab}\}$ constants satisfying $J^{\,ba}=J^{\,ab}$,
and $\psi$ is a strictly convex function 
( See Refs.\,\cite{Hopfield1984,Cohen1983} ). 
Recall that the sum of 
strictly convex functions is also a strictly convex function.  
Note that all the constants $\{C_{\,i}\}$ appeared in 
Ref.\,\cite{Hopfield1984} 
have been set to unity in this paper. 
A Lyapunov function is known to exist in this model. 
Choose $\cH$ to be 
\beq
\cH(V)
=-\,\sum_{a}\left[\frac{1}{2}\sum_{b}
J^{\,ab}V_{\,a}V_{\,b}
-\int_{0}^{V_{\,a}}\frac{1}{R_{\,a}}\Upsilon^{\,-1}(V^{\,\prime})
\,\dr V^{\,\prime}
+V_{\,a}I_{\,\ext}^{\,a}\,\right],
\label{example-Cohen-Grossberg-H}
\eeq
where we have assumed that $\Upsilon^{-1}$ exists. Then it follows from 
$$
\frac{\partial\cH}{\partial V_{\,a}}
=-\left[\,\sum_{b}J^{\,ab}\,V_{\,b}-\frac{1}{R_{\,a}}U^{\,a}
+I_{\,\ext}^{\,a}\,\right]
=-\,\frac{\dr}{\dr t}U^{\,a}
$$
that 
$$
\frac{\dr\cH(V)}{\dr t}
=\sum_{a}\frac{\partial \cH}{\partial V_{\,a}}\frac{\dr V_{\,a}}{\dr t}
=\sum_{a}\frac{\partial \cH}{\partial V_{\,a}}
\frac{\dr \Upsilon}{\dr U^{\,a}}
\frac{\dr U^{\,a}}{\dr t}
=-\,\sum_{a}\frac{\dr U^{\,a}}{\dr t}
\frac{\dr \Upsilon}{\dr U^{\,a}}
\frac{\dr U^{\,a}}{\dr t}
<0.
$$
Thus the generalized 
Hopfield 
model includes the 
 Hopfield
model.
%%%%%%%%%%%%%
\end{Example}
%%%%%%%%%%%%%%
%%%%%%%%%%%%%%%%%%%%%%%%%%%%%%%%%%%%%%
\subsection{Geometric description}
%%%%%%%%%%%%%%%%%%%%%%%%%%%%%%%%%%%%%%%%
A geometric description of the generalized Hopfield model is described in this 
subsection.  
First, how the strictly convex function $\Psi$ induces a Hessian manifold 
is shown.
Hessian manifold is a triplet $(\cM,\nabla^{\,\Psi},g^{\,\Psi})$, where 
$\nabla^{\,\Psi}$ is a flat connection, and $g^{\,\Psi}$  the Riemannian metric 
tensor field satisfying $g^{\,\Psi}=\nabla^{\,\Psi}\dr\Psi$. 
The connection is such that 
$\nabla_{\partial/\partial U^{\,a}}^{\,\Psi}(\partial/\partial U^{\,b})=0$ 
for a local coordinate set $U$. 
 
Given $\Psi$, define the Riemannian metric tensor field as 
\beq
g^{\,\Psi}
=\sum_{a}\sum_{b}\,g_{\,ab}^{\,\Psi}\,
\dr U^{\,a}\otimes\dr U^{\,b},
\qquad\mbox{where}\qquad 
g_{\,ab}^{\,\Psi}
=\frac{\partial^{2}\,\Psi}{\partial U^{\,a}\partial U^{\,b}}.
\label{metric-psi}
\eeq
As in \fr{V-U-psi}, the function $\Psi$ of $U$ induces $V$. 

With \fr{V-U-psi} and \fr{metric-psi}, one has
$$
g_{\,ab}^{\,\Psi}
=\frac{\partial V_{\,a}}{\partial U^{\,b}}
=\frac{\partial^{2}\,\Psi}{\partial U^{\,a}\partial U^{\,b}},
$$
which form the matrix $(g_{\,ab}^{\,\Psi})$. 
The inverse matrix of $(g_{\,ab}^{\,\Psi})$, denoted by $(g_{\,\Psi}^{\,ab})$, 
is known to be written 
as \cite{AN}
$$
g_{\,\Psi}^{\,ab}
=\frac{\partial U^{\,a}}{\partial V_{\,b}}
=\frac{\partial^{\,2}\Psi^{\,*}}{\partial V_{\,a}\partial V_{\,b}},
$$
where $\Psi^{\,*}$ is the total Legendre transform of $\Psi$:
$$
\Psi^{\,*}(V)
:=\left[\sum_{a}U^{\,a}V_{\,a}-\Psi(U)\right]_{U=U(V)}.
$$
Combining these, one has 
\beq
g^{\,\Psi}
=\sum_{a}\sum_{b}g_{\,ab}^{\,\Psi}\,\dr U^{\,a}\otimes\dr U^{\,b}
=\sum_{a}\sum_{b}g_{\,\Psi}^{\,ab}\,\dr V_{\,a}\otimes\dr V_{\,b}.
\label{metric-psi-U-V}
\eeq
The coordinates $U$ and $V$ are dual in the sense of information 
geometry \cite{AN} :  
$$
g^{\,\Psi}\left(\,
\frac{\partial}{\partial U^{\,a}},
\frac{\partial}{\partial V_{\,b}}
\,\right)
=\delta_{\,a}^{\,b}.
$$
From $(\cM,\nabla^{\,\Psi},g^{\,\Psi})$, one can uniquely 
introduce another connection 
denoted $\nabla^{\,\Psi\,*}$ that satisfies 
$$
X[g^{\,\Psi}(Y,Z)]
=g^{\,\Psi}(\nabla_{\,X}^{\,\Psi}Y,Z)+
g^{\,\Psi}(Y,\nabla_{\,X}^{\,\Psi\,*}Z), 
$$ 
for all vector fields $X,Y$ and $Z$.  
Then it turns out that 
the connection $\nabla^{\,\Psi\,*}$ is such that 
$\nabla_{\partial/\partial V_{\,a}}^{\,\Psi\,*}{\partial/\partial V_{\,b}}=0$. 
A quadruplet $(\cM,g^{\,\Psi},\nabla^{\,\Psi},\nabla^{\,\Psi\,*})$ 
is referred to as a dually flat space \cite{AN}.  
Since $(\cM,\nabla^{\,\Psi},g^{\,\Psi})$ is a type of Riemannian manifold 
$(\cM,g^{\,\Psi})$ induced from $\Psi$, 
a canonical volume-form is defined. To emphasize 
how this volume-form is induced,  
this volume-form is denoted $\star_{\,\Psi}1$ in this paper.
Associated with $\star_{\,\Psi}1$, the Hodge map 
$\star_{\,\Psi}:\GamLamM{p}\to\GamLamM{n-p}$ is defined, where 
$\GamLamM{p}$ is the space of $p$-forms on $\cM$.

An expression for the generalized 
Hopfield 
 model 
is written in terms of $g^{\,\Psi}$ as follows.

%%%%%%%%%%%%%
\begin{Lemma}
\label{main-lemma}
%%%%%%%%%%%%
The generalized 
Hopfield 
model is written as the components of 
a vector field $X_{\,\cH}$ on 
the Riemannian manifold $(\cM,g^{\,\Psi})$ satisfying
\beq
g^{\,\Psi}(X_{\,\cH},-)
=-\,\dr \cH,
\label{geometric-generalized-Cohen-Grossberg-model-1}
\eeq 
where $g^{\,\Psi}$ has been defined in \fr{metric-psi}. 
%%%%%%%%%%%%
\end{Lemma}
%%%%%%%%%%%%
%%%%%%%%%%%%%
\begin{Proof}
%%%%%%%%%%%%
Write $X_{\,\cH}$ in terms of $\{U^{\,a}\}$ as  
$$
X_{\,\cH}
=\sum_{a}\dot{U}^{\,a}\frac{\partial}{\partial U^{\,a}},\qquad\mbox{with}\qquad
\dot{U}^{\,a}
=\frac{\dr U^{\,a}}{\dr t}.
$$
Then, substituting
$$
g^{\,\Psi}(X_{\,\cH},-)
=\sum_{a}\sum_{b}
g_{\,ab}^{\,\Psi}\,\dot{U}^{\,a}\dr U^{\,b},\quad\mbox{and}\quad
\dr\,\cH
=\sum_{a}\sum_{b}\frac{\partial\cH}{\partial V_{\,a}}
\frac{\dr V_{\,a}}{\dr U^{\,b}}
\dr U^{\,b}
=\sum_{a}\sum_{b}g_{\,ab}^{\,\Psi}\frac{\partial\cH}{\partial V_{\,a}}
\dr U^{\,b},
$$
into \fr{geometric-generalized-Cohen-Grossberg-model-1}, using the property of  
the pairing $\dr U^{\,b}(\partial/\partial U_{\,a})=\delta^{\,ab}$, and 
$g_{\,ab}^{\,\Psi}=g_{\,b a}^{\,\Psi}$, one has
$$
\sum_{b}g_{\,ab}^{\,\Psi}\frac{\dr U^{\,a}}{\dr t}
=-\sum_{b}\,g_{\,ab}^{\,\Psi}\,\frac{\partial\cH}{\partial V_{\,b}}.
$$
Since $g^{\,\Psi}$ is non-degenerate, one has 
$$
\frac{\dr U^{\,a}}{\dr t}
=-\,\frac{\partial\cH}{\partial V_{\,a}}.
$$
\qed
%%%%%%%%%%%%
\end{Proof}
%%%%%%%%%%
From this Lemma, a generalized  
Hopfield 
model is a triplet $(\cM,\Psi,\cH)$  
in this geometric setting, and can be viewed as a dynamical system on the 
dually flat space $(\cM,g^{\,\Psi},\nabla^{\,\Psi},\nabla^{\,\Psi\,*})$.
In the literature,  several dynamical systems have been studied
in dually flat spaces \cite{Nakamura1993,Fujiwara-Amari1995,Goto2016}. 
In the so-called statistical manifolds, which are manifolds generalized from 
dually flat spaces,   
several dynamical systems theories have been 
considered \cite{Noda2011,Leok2017,Goto2018}.

This Lemma will be used to calculate the phase space compressibility.
This quantity is associated with a volume-form, and expresses how 
such a volume-form is compressed along a vector field associated with 
a dynamical system. If there is an attractor in phase space for 
a dynamical system, then   
roughly speaking, phase space compressibility is  a measure for expressing 
how fast flow of a dynamical system converges to an attracting set in the 
phase space. This is defined as follows:
%%%%%%%%%%%%
\begin{Def}
%%%%%%%%%%
( Phase space compressibility \cite{Ezra2002} ) :  
Let $\cM$ 
be an $n$-dimensional manifold, $X$ a vector field on $\cM$, 
and $\Omega$ a non-vanishing $n$-form. Introduce the one-form 
$\kappa_{\,\Omega}$ such that
$\cL_{X}\Omega = \kappa_{\Omega}(X) \Omega$, where 
$\cL_{X}$ is the Lie derivative along a vector field $X$. 
Then $\kappa_{\,\Omega}(X)$ 
is referred to as a phase space compressibility.
%%%%%%%%%
\end{Def}
%%%%%%%%%
Examples of phase space compressibility are as follows.
%%%%%%%%%%%%%%%%%%%
\begin{enumerate}
%%%%%%%%%%%%%%%%%%5
%%%%%%%%%%%%%%%%%%5
\item
%%%%%%%
(Vector field associated with a linear dynamical system):
Consider the vector field $X$ on $\mbbR^{\,n}$ associated with the 
linear dynamical system, 
$$
X=\sum_{a=1}^{n}\dot{q}^{\,q}\frac{\partial}{\partial q^{\,a}},
\qquad
\dot{q}^{\,a}
=-\,\check{\kappa}^{\,a}\,q^{\,q},\qquad a\in\{1,\ldots,n\}
$$ 
 where $\check{\kappa}^{\,a}>0$ is constant for each $a$. Choose 
$\Omega=\dr q^{\,1}\wedge\cdots\wedge \dr q^{\,n}$ as a non-vanishing $n$-form.
It follows from 
$$
\cL_{X}\Omega
=\left[\,(\cL_{X}\dr q^{\,1})\wedge \cdots\wedge
\dr q^{\,n}\,\right]+\cdots
+\left[\,\dr q^{\,1}\wedge\cdots\wedge(\cL_{X}\dr q^{\,n})\,\right], 
$$ 
and 
$\cL_{X}\dr q^{\,a}=\dr (Xq^{\,a})=\dr \dot{q}^{\,a}=-\check{\kappa}^{\,a}\dr q^{\,a}$
that 
$$
\cL_{\,X}\Omega
=-\,\left(\,\sum_{a=1}^{n}\check{\kappa}^{\,a}\,\right)\,\Omega,
$$
from which 
$\kappa_{\Omega}(X)=-\,\sum_{a=1}^{n}\check{\kappa}^{\,a}<0$.
%%%%%%
\item
%%%%%%%
( Hamiltonian vector field ) :   
Let $(\cS,\omega)$ be an $2n$-dimensional symplectic manifold, $(q,p)$
a Darboux coordinate system so that 
$\omega=\sum_{a=1}^{n}\dr p_{\,a}\wedge \dr q^{\,a}$ with 
 $q=\{q_{\,1},\ldots,q_{\,n}\}$ and $p=\{p^{\,1},\ldots,p^{\,n}\}$, 
and $H$ a Hamiltonian function. The Hamiltonian vector field $X_{\,H}$ is 
the vector field satisfying $\ii_{X_{\,H}}\omega=-\,\dr H$. It then follows that 
$\cL_{X_{\,H}}\omega=0$. Choose 
$\Omega=\omega\wedge\cdots\wedge\omega$ as a non-vanishing $2n$-form on 
$\cS$. 
Since $\cL_{X_{\,H}}\Omega=0$, one concludes that $\kappa_{\Omega}(X_{\,H})=0$.
%%%%%%%%%%%%%%%%5
\end{enumerate}
%%%%%%%%%%%%%%%%%
Recall that on a Riemannian manifold $(\cM,g)$, the (Hodge) Laplace 
operator acting on a function is defined as 
$\star^{\,-1}\,\dr\,\star\,\dr:\GamLamM{0}\to\GamLamM{0}$, where $\star$ is 
the Hodge operator and $\star^{\,-1}$ its inverse on $\cM$. 
Then, the main theorem in this paper is as follows.
%%%%%%%%%%%%
\begin{Thm}
\label{fact-phase-space-compressibility}
%%%%%%%%%%
( Phase space compressibility for the generalized 
Hopfield model ) : 
Let $\star_{\,\Psi} 1$ be a canonical volume-form on $(\cM,g^{\,\Psi})$, 
$\star_{\,\Psi}$ the Hodge map, $\star_{\,\Psi}^{\,-1}$ its inverse map, 
and $X_{\,\cH}$ a vector field 
satisfying \fr{geometric-generalized-Cohen-Grossberg-model-1}.
Then, the phase space compressibility 
$\kappa_{\,\Psi}^{\,\cH}:=\kappa_{\,\star_{\,\Psi} 1}(X_{\,\cH})$ 
is given by the negative of the Laplace operator acting on the function $\cH$, 
$$
\kappa_{\,\Psi}^{\,\cH}
=-\,\star_{\,\Psi}^{-1}\dr \star_{\,\Psi}\dr \cH.
$$  
%%%%%%%%%5
\end{Thm} 
%%%%%%%%%%
%%%%%%%%%5
\begin{Proof} 
%%%%%%%%%%
Introduce the notation $\wt{X}_{\,\cH}=g^{\,\Psi}(X_{\,\cH},-)$, which is 
the metric dual of $X_{\,\cH}$.
To proceed, we use the formula 
$$
\ii_X\star 1
=\star\, \wt{X},
$$
for any vector field $X$ on a Riemannian manifold, 
where $\star$ is a Hodge map and 
$\ii_{\,X}$ the interior product operator with a vector field $X$. 
Then, 
it follows from \fr{geometric-generalized-Cohen-Grossberg-model-1} and 
the formula that 
$\ii_{X_{\,\cH}}\star_{\,\Psi}1=\star_{\,\Psi}\, \wt{X}_{\,\cH}=-\star_{\,\Psi}\dr \cH$. 
Also, from definition of $\kappa_{\,\Psi}^{\,\cH}$, it follows that  
$\kappa_{\,\Psi}^{\,\cH}\star_{\,\Psi} 1=\cL_{\,X_{\,\cH}}\star_{\,\Psi} 1$.

With these and the Cartan formula $\cL_{X}\beta=(\dr\ii_{X}+\ii_X\dr)\beta$ for 
any $p$-form $\beta$ with $0\leq p\leq \dim\cM$, one has that 
$$
\kappa_{\,\Psi}^{\,\cH}
=\star_{\,\Psi}^{-1}(\,\kappa_{\,\Psi}^{\,\cH}\star_{\,\Psi} 1\,)
=\star_{\,\Psi}^{-1}(\,\cL_{\,X_{\,\cH}}\star_{\,\Psi} 1\,)
=\star_{\,\Psi}^{-1}\,(\dr \ii_{\,X_{\,\cH}}\star_{\,\Psi} 1\,)
=\star_{\,\Psi}^{-1}\,\left[\,\dr\, (-\star_{\,\Psi}\dr \cH\,)\,\right]
=-\,\star_{\,\Psi}^{-1}\dr \star_{\,\Psi}\dr \cH.
$$
\qed
%%%%%%%%%5
\end{Proof} 
%%%%%%%%%% 
There are some consequences from Lemma\,\ref{main-lemma} and  
Theorem\,\ref{fact-phase-space-compressibility}. 
%%%%%%%%%%%%%%%%% 
\begin{itemize}
%%%%%%
\item
%%%%%%%
It follows from \fr{geometric-generalized-Cohen-Grossberg-model-1} that
\beq
\dr\wt{X}_{\,\cH}
=0.
\label{divergenceless-XH}
\eeq
%%%%%%
\item
%%%%%%%
For the case that $\partial\cM=\emptyset$, it follows from the Stokes theorem 
that 
$$
\int_{\cM}\kappa_{\,\Psi}^{\,\cH}\star_{\,\Psi}1
=0.
$$
%%%%%%
\item
%%%%%5
The existence of Lyapunov function can be expressed in terms of the metric 
tensor field as 
$$
\frac{\dr\cH}{\dr t}
=X_{\,\cH}\cH
=\dr\cH(X_{\,\cH})
=-\,g^{\,\Psi}(X_{\,\cH},X_{\,\cH})
<0,
$$
for a non-vanishing vector field $X_{\,\cH}$. 
%%%%%%
\item
%%%%%5
In terms of the coordinate set $V$, 
$$
X_{\,\cH}
=\sum_a\dot{V}_{\,a}\frac{\partial}{\partial V_{\,a}},
$$
one has 
$$
\sum_{b}g_{\,\Psi}^{\,ab}\dot{V}_{\,b}
=-\,\frac{\partial\cH}{\partial V_{\,a}}.
$$
%%%%%%%%%%%%%%
\end{itemize}
%%%%%%%%%%%%%%%

The phase space compressibility $\kappa_{\,\Psi}^{\cH}$ 
is also related to the co-derivative and 
the (Hodge) Laplace operator acting on the one-form 
$\wt{X}_{\,\cH}$. To state these explicitly, recall that 
the co-derivative acting on a $p$-form 
$\dr_{\,\Psi}^{\,\dagger}:\GamLamM{p}\to\GamLamM{p-1}$ 
is defined as 
$\dr_{\,\Psi}^{\,\dagger}=\star^{-1}\,\dr\,\star$, and that 
the Laplace operator acting on a $p$-from on $\cM$ is defined as 
$\dr\,\dr^{\,\dagger}+\dr^{\,\dagger}\,\dr:\GamLamM{p}\to\GamLamM{p}$, 
( See for example, Ref.\,\cite{Nakahara1990} ).
Then, one has the following. 

%%%%%%%%%%%%%%%%%%%%%%
\begin{Proposition}
%%%%%%%%%%%%%%%%%%%%
$$
\dr_{\,\Psi}^{\,\dagger}\,\wt{X}_{\,\cH}
=\kappa_{\,\Psi}^{\,\cH},\qquad\mbox{and}\qquad
\qquad
(\,\dr\,\dr_{\,\Psi}^{\,\dagger}+
\dr_{\,\Psi}^{\,\dagger}\,\dr \,)\,\wt{X}_{\,\cH}
=\dr\kappa_{\,\Psi}^{\,\cH},
$$
where $\dr_{\,\Psi}^{\,\dagger}:\GamLamM{p}\to\GamLamM{p-1}$    
is the co-derivative, 
$\dr_{\,\Psi}^{\,\dagger}:=\star_{\Psi}^{-1}\dr\star_{\,\Psi}$. 
%%%%%%%%%%%%%%%%%%%
\end{Proposition}
%%%%%%%%%%%%%%%%%%%
\begin{proof}
%%%%%%%%%%%%%%
A proof is completed by straightforward calculations.   
For the first equality, 
it follows from $\wt{X}_{\,\cH}=-\,\dr \cH$ that 
$$
\dr_{\,\Psi}^{\,\dagger}\,\wt{X}_{\,\cH}
=\star_{\,\Psi}^{-1}\dr\star_{\,\Psi}\,\wt{X}_{\,\cH}
=-\,\star_{\,\Psi}^{-1}\dr\star_{\,\Psi}\,\dr \cH\,
$$
The most right hand side of the equation above 
is written in terms of $\kappa_{\,\Psi}^{\,\cH}$ due to 
Theorem\,\ref{fact-phase-space-compressibility} :  
$$
\dr_{\,\Psi}^{\,\dagger}\,\wt{X}_{\,\cH}
=\kappa_{\,\Psi}^{\,\cH}.
$$
For the second equality it follows from  \fr{divergenceless-XH} and 
the proven first equality that 
$$
(\,\dr\,\dr_{\,\Psi}^{\,\dagger}+
\dr_{\,\Psi}^{\,\dagger}\,\dr \,)\,\wt{X}_{\,\cH}
=\dr\,\dr_{\,\Psi}^{\,\dagger}
\,\wt{X}_{\,\cH}
=\dr\kappa_{\,\Psi}^{\,\cH}.
$$
%%%%%%%%%%%%
\end{proof}
%%%%%%%%%%%%%

Before closing this subsection, 
it is argued how this geometric formulation is applied to other 
dynamical systems. 
It has been known that the Hopfield model belongs  a class of 
dynamical systems  proposed in Ref.\,\cite{Cohen1983} 
( See also \cite{Grossberg1988} ). 
Consider the dynamical system of the form,
\beq
\frac{\dr U^{\,a}}{\dr t}
=\left(B^{\,a}-\sum_{j=1}^{n}C^{\,aj}\frac{\dr\psi}{\dr U^{\,j}}\right)
\,A^{\,a},
\qquad a\in\{1,\ldots,n\}
\label{restricted-Cohen-Grossberg}
\eeq 
where
%%%%%%%%%%%%%%%% 
\begin{itemize}
\item
%%%%%%%
$A^{\,a}$ is a positive function of $U^{\,a}$, $A^{\,a}=A^{\,a}(U^{\,a})$,  
%%%%%%
\item
%%%%%
$B^{\,a}$ is a function of $U^{\,a}$, $B^{\,a}=B^{\,a}(U^{\,a})$,
%%%%%%
\item
%%%%%
$C^{\,ab}$ forms a symmetric constant matrix, $C^{\,ab}=C^{\,ba}$,  
%%%%%%
\item
%%%%%
$\psi$ is a strictly convex function depending on one variable, so that  
$\dr^{\,2}\psi/\dr \xi^{\,2}>0$.
%%%%%%%%%%%%%
\end{itemize}
%%%%%%%%%%%%%%
The class of the form \fr{restricted-Cohen-Grossberg} 
is obtained by restricting the class proposed in Ref.\,\cite{Cohen1983}.
The Lyapunov function $\cH^{\,\prime}$ was found as  
$$
\cH^{\,\prime}(U)
=-\sum_{j=1}^{n}\int^{U^{\,j}}B^{\,j}(\xi)\,
\frac{\dr^2\psi}{\dr \xi^{\,2}}\dr \xi
+\frac{1}{2}\sum_{j=1}^{n}\sum_{k=1}^{n}
C^{\,jk}\frac{\dr\psi}{\dr U^{\,j}}\frac{\dr\psi}{\dr U^{\,k}}.
$$
From this $\cH^{\,\prime}$, one has 
$$
\frac{\partial\cH^{\,\prime}}{\partial U^{\,j}}
=\left(\,-B^{\,j}+\sum_{k=1}^{n}C^{\,jk}\frac{\dr\psi}{\dr U^{\,k}}
\,\right)\,\frac{\dr^{2}\psi}{\dr U^{\,j\,2}}
=-\,\frac{\dr^{2} \psi}{\dr U^{\,j\,2}}\frac{1}{A^{\,j}}\frac{\dr U^{\,j}}{\dr t},
$$
from which 
$$
\frac{\dr\cH^{\,\prime}}{\dr t}
=\sum_{j=1}^{n}\frac{\dr \cH^{\,\prime}}{\dr U^{\,j}}\frac{\dr U^{\,j}}{\dr t}
=-\,\sum_{j=1}^{n}\frac{\dr^{2} \psi}{\dr U^{\,j\,2}}\frac{1}{A^{\,j}}
\left(\frac{\dr U^{\,j}}{\dr t}\right)^{\,2}
<0.
$$
Then, \fr{restricted-Cohen-Grossberg} can be written as 
\beq 
\frac{\dr U^{\,a}}{\dr t}
=-\,A^{\,a}\frac{\partial \cH^{\,\prime}}{\partial V_{\,a}},\qquad
\mbox{with}\qquad V_{\,a}
=\frac{\dr \psi}{\dr U^{\,a}}, 
\label{restricted-Cohen-Grossberg-with-Lyapunov}
\eeq
due to 
$$
\frac{\partial \cH^{\,\prime}}{\partial V_{\,a}}
=\sum_{j=1}^{n}\frac{\partial\cH^{\,\prime}}{\partial U^{\,j}}
\frac{\dr U^{\,j}}{\dr V_{\,a}}
=-\sum_{j=1}^{n}\frac{\dr^{2} \psi}{\dr U^{\,j\,2}}\frac{1}{A^{\,j}}
\frac{\dr U^{\,j}}{\dr t}\delta^{\,ja}
\left(\frac{\dr^2\psi}{\dr U^{\,j\,2}}\,\right)^{\,-1}
=-\,\frac{1}{A^{\,a}}\frac{\dr U^{\,a}}{\dr t}.
$$
Notice that \fr{restricted-Cohen-Grossberg-with-Lyapunov}
is not written as  
\fr{generalized-Cohen-Grossberg-model-H-differential}, unless $A^{\,a}$  
is constant.

%%%%%%%%%%%%%%%%%%%%%%%%
\subsection{Coordinate expression of phase space compressibility}
%%%%%%%%%%%%%%%%%%%%%%%%%%%
In this subsection the coordinate expression of $\kappa_{\,\Psi}^{\,\cH}$ 
is given. In particular the case 
\beq
\Psi(U)
=\sum_{a}\psi(U^{\,a}),
\label{Psi-sum-psi}
\eeq
is focused.

Recall that the coordinate expression of the 
Laplace operator on a Riemannian manifold $(\cN,g)$  
with $g$ being a Riemannian metric tensor field
$$
g=\sum_{a}\sum_{b}g_{\,ab}\,\dr x^{\,a}\otimes\dr x^{\,b}.
$$
The Laplacian operator acting on a function $f$ is then 
\beq
\star^{-1}\dr \star \dr f
=\sum_{a}\sum_{b}\frac{1}{\sqrt{|g|}}\frac{\partial}{\partial x^{\,a}}
\left(\sqrt{|g|}\,g^{\,ab}\frac{\partial f}{\partial x^{\,b}}\right)
=\sum_{a}\sum_{b}\left[\,\frac{\partial(\ln\sqrt{|g|})}{\partial x^{\,a}}
g^{\,ab}\frac{\partial f}{\partial x^{\,b}}
+\frac{\partial}{\partial x^{\,a}}\left(\,
g^{\,ab}\frac{\partial f}{\partial x^{\,b}}
\,\right)\,\right],
\label{Laplacian-function-formula}
\eeq
where $\star$ is the Hodge map associated with $g$, 
$|g|$ the determinant of the matrix $(g_{\,ab})$, and 
$(g^{\,ab})$ the inverse matrix of $(g_{\,ab})$. 

A coordinate expression of $\kappa_{\,\Psi}^{\,\cH}$ 
is calculated by applying \fr{Laplacian-function-formula} 
to Theorem\ref{fact-phase-space-compressibility}.
Each term in \fr{Laplacian-function-formula} is calculated as follows. 
Since 
$$
g_{\,ab}^{\,\Psi}
=\frac{\dr^{2}\psi}{\dr U^{\,a\,2}}\delta_{\,ab},
\qquad
|\,g^{\,\Psi}\,|
=\prod_{a=1}^n\frac{\dr^2\,\psi}{\dr U^{\,a\,2}},\qquad 
\mbox{and}\qquad 
g_{\,\Psi}^{\,ab}
=\delta^{\,ab}\frac{1}{\frac{\dr^2\,\psi}{\dr U^{\,a\,2}}},
$$
one has 
\beqa
\frac{\partial(\,\ln\sqrt{|g^{\,\Psi}|}\,)}{\partial U^{\,a}}
&=&\frac{1}{2}\frac{\partial}{\partial U^{\,a}}\sum_{b=1}^n\ln\left(\,
\frac{\dr^2\,\psi}{\dr U^{\,b\,2}}
\,\right)
=\frac{1}{2}\frac{1}{\frac{\dr^2\,\psi}{\dr U^{\,a\,2}}}
\frac{\dr^3\,\psi}{\dr U^{\,a\,3}},
\non\\
\frac{\partial(\,\ln\sqrt{|g^{\,\Psi}|}\,)}{\partial U^{\,a}}
\,g_{\,\Psi}^{\,ab}\,
\frac{\partial\,\cH}{\partial U^{\,b}}
&=&
=\frac{1}{2}\delta^{\,ab}\frac{1}{\left(\frac{\dr^2\,\psi}{\dr U^{\,a\,2}}\,\right)^{\,2}}
\frac{\dr^3\,\psi}{\dr U^{\,a\,3}}
\frac{\partial\,\cH}{\partial U^{\,a}},
\non
\eeqa  
and
$$
\frac{\partial}{\partial U^{\,a}}\left(\,
g_{\,\Psi}^{\,ab}\frac{\partial\cH}{\partial U^{\,b}}\,\right)
=\delta^{\,ab}\frac{\partial}{\partial U^{\,a}}\left(\,
\frac{1}{\frac{\dr^2\,\psi}{\dr U^{\,a\,2}}}
\frac{\partial\cH}{\partial U^{\,b}}
\,\right)
=\delta^{\,ab}\left[\,
-\,\frac{\frac{\dr^3\,\psi}{\dr x^{\,a\,3}}}{\left(\frac{\dr^2\,\psi}{\dr U^{\,a\,2}}\right)^2}
\frac{\partial\cH}{\partial U^{\,b}}
+\frac{1}{\frac{\dr^2\,\psi}{\dr U^{\,a\,2}}}
\frac{\partial^2\cH}{\partial U^{\,b}\partial U^{\,a}}\,\right].
$$
Combining these terms, one can write   
$$ 
\kappa_{\,\Psi}^{\,\cH}
=-\,\sum_{a}\frac{1}{\left(\frac{\dr^2\psi}{\dr U^{\,a\,2}}\right)^2}
\left[\,
-\,\frac{1}{2}\frac{\dr^3\psi}{\dr U^{\,a\,3}}
\frac{\partial\cH}{\partial U^{\,a}}
+\frac{\dr^2\psi}{\dr U^{\,a\,2}}
\frac{\partial^2\cH}{\partial U^{\,a\,2}}\,\right].
$$
Since $\cH=\cH(V)$, one has 
$$
\frac{\partial\cH}{\partial U^{\,a}}
=\sum_{b}\frac{\partial\cH}{\partial V_{\,b}}
\frac{\partial V_{\,b}}{\partial U^{\,a}}
=g_{\,aa}^{\,\Psi}\frac{\partial\cH}{\partial V_{\,a}}
=\frac{\dr^2\,\psi}{\dr U^{\,a\,2}}\frac{\partial\cH}{\partial V_{\,a}},
$$
and 
\beqa
\frac{\partial^2\cH}{\partial U^{\,a\,2}}
&=&\frac{\partial}{\partial U^{\,a}}\left(\,
g_{\,aa}^{\,\Psi}\frac{\partial\cH}{\partial V_{\,a}}\,\right)
=\frac{\partial g_{\,aa}^{\,\Psi}}{\partial U^{\,a}}
\frac{\partial\cH}{\partial V_{\,a}}
+ g_{\,aa}^{\,\Psi}\frac{\partial}{\partial U^{\,a}}\left(
\frac{\partial\cH}{\partial V_{\,a}}
\right)
=\frac{\partial g_{\,aa}^{\,\Psi}}{\partial U^{\,a}}
\frac{\partial\cH}{\partial V_{\,a}}
+ \sum_bg_{\,aa}^{\,\Psi}g_{\,ab}^{\,\Psi}
\frac{\partial^2\cH}{\partial V_{\,b}\partial V_{\,a}}
\non\\
&=&\frac{\partial g_{\,aa}^{\,\Psi}}{\partial U^{\,a}}
\frac{\partial\cH}{\partial V_{\,a}}
+ (\,g_{\,aa}^{\,\Psi}\,)^{\,2}
\frac{\partial^2\cH}{\partial V_{\,a}^{\,2}}
=\frac{\dr^3\psi}{\dr U^{\,a\,3}}
\frac{\partial\cH}{\partial V_{\,a}}
+ \left(\,\frac{\dr^2\psi}{\dr U^{\,a\,2}}\,\right)^{\,2}
\frac{\partial^2\cH}{\partial V_{\,a}^{\,2}}.
\non
\eeqa
Thus, one arrives at 
\beqa
\kappa_{\,\Psi}^{\,\cH}
&=&-\,\sum_{a}\frac{1}{\left(\frac{\dr^2\psi}{\dr U^{\,a\,2}}\right)^2}
\left[\,
-\,\frac{1}{2}\frac{\dr^3\psi}{\dr U^{\,a\,3}}
\left(\frac{\dr^2\,\psi}{\dr U^{\,a\,2}}\frac{\partial\cH}{\partial V_{\,a}}
\right)
+\frac{\dr^2\psi}{\dr U^{\,a\,2}}
\left(\frac{\dr^3\psi}{\dr U^{\,a\,3}}
\frac{\partial\cH}{\partial V_{\,a}}
+ \left(\,\frac{\dr^2\psi}{\dr U^{\,a\,2}}\,\right)^{\,2}
\frac{\partial^2\cH}{\partial V_{\,a}^{\,2}}
\right)
\,\right]
\non\\
&=&
-\,\sum_{a}\frac{1}{\frac{\dr^2\psi}{\dr U^{\,a\,2}}}
\left[\,
\frac{1}{2}\frac{\dr^3\psi}{\dr U^{\,a\,3}}
\frac{\partial\cH}{\partial V_{\,a}}
+ \left(\,\frac{\dr^2\psi}{\dr U^{\,a\,2}}\,\right)^{\,2}
\frac{\partial^2\cH}{\partial V_{\,a}^{\,2}}
\,\right].
\label{phase-space-compressibility-homogeneous-case}
\eeqa

%%%%%%%%%%%%%%%%%%%%%
\section{Examples}
\label{section-examples}
%%%%%%%%%%%%%%%%%%%%%%%%%%5
In this section, after introducing the sigmoid function,  
the two examples are focused. 

As an activation function, \fr{Psi-sum-psi} 
is often focused in the literature due to simplicity, 
$$ 
\Psi(U)
=\sum_{a}\psi(U^{\,a}).
$$
It follows that the matrix $(g_{\,ab}^{\,\Psi})$ is diagonal. 

In what follows geometric objects and 
the phase space compressibility are calculated for this case. 

%%%%%%%%%%%%%%%%
\subsection{Sigmoid function}
%%%%%%%%%%%%%%%%5
Choose the function $\psi$ as 
$$
\psi(x)
=\ln(\,1+\e^{\,x}\,).
$$
This is referred to as the soft plus function, and this 
choice leads to the sigmoid function as shown below. 
From this  $\psi$, one can derive various quantities.
First introduce 
$$
x^{\,*}
:=\frac{\dr\psi}{\dr x}
=\frac{\e^{\,x}}{1+\e^{\,x}},
$$
where the right hand side of the equation above 
is referred to as the sigmoid function.
Then,  
$$
\frac{\dr x^{\,*}}{\dr x}
=\frac{\dr^2\psi}{\dr x^{\,2}}
=\frac{\e^{\,x}}{(1+\e^{\,x})^2}
=x^{\,*}(1-x^{\,*}),
$$
and 
$$
\frac{\dr^2 x^{\,*}}{\dr x^2}
=\frac{\dr^3\psi}{\dr x^{\,3}}
=\frac{\e^{\,x}(1-\e^{\,x})}{(1+\e^{\,x})^3}
=\frac{\dr x^{\,*}}{\dr x}(1-2x^{\,*})
=x^{\,*}(1-x^{\,*})(1-2x^{\,*}).
$$

The Legendre transform of $\psi$,  
$$
\psi^{\,*}(x^{\,*})
=\left[\,xx^{\,*}-\psi(x)\,\right]_{x=x(x^{\,*})},
$$
is calculated as follows. Solving 
$$
x^{\,*}
=\frac{\dr\psi}{\dr x}
=\frac{\e^{\,x}}{1+\e^{\,x}},\qquad 
(\,0<x^{\,*}<1\,)
$$
for $x$, one has
$$
x=\ln\left(\frac{x^{\,*}}{1-x^{\,*}}\right),\qquad
\mbox{and}\qquad 
\e^{\,x}
=\frac{x^{\,*}}{1-x^{\,*}}.
$$
Then,   
$$
\psi^{\,*}(x^{\,*})
=x^{\,*}\ln x^{\,*}+(1-x^{\,*})\ln(1-x^{\,*}),
$$
from which 
$$
x=\frac{\dr\psi^{\,*}}{\dr x^{\,*}}
=\ln x^{\,*}-\ln\left(1-x^{\,*}\right),\qquad\mbox{and}\qquad
\frac{\dr x}{\dr x^{\,*}}
=\frac{1}{x^{\,*}}+\frac{1}{1-x^{\,*}}.
$$

The components of the metric tensor field are obtained as 
$$
g_{\,ab}^{\,\Psi}
=\delta_{\,ab}\frac{\e^{\,U^{\,a}}}{(1+\e^{\,U^{\,a}})^2}
=\delta_{\,ab}V_{\,a}(1-V_{\,a}),\qquad
\mbox{and}\qquad
g_{\,\Psi}^{\,ab}
=\delta^{\,ab}\frac{(1+\e^{\,U^{\,a}})^2}{\e^{\,U^{\,a}}}
=\delta^{\,ab}\left(
\frac{1}{V_{\,a}}+\frac{1}{1-V_{\,a}}
\right).
$$

%%%%%%%%%%%%%%%%%%
\subsection{Example\,\ref{example-simplest}}
%%%%%%%%%%%%%%%%%%%%%%
The case of Example\,\ref{example-simplest} 
in Section\,\ref{section-model} is considered here, where 
$\cH(V)=\sum_a\sum_b\delta^{\,ab}\,V_{\,a}V_{\,b}/2$ has been chosen.
 
Substituting this $\cH$ and 
$$
\frac{\partial\cH}{\partial V_{\,a}}
=\sum_{b}\delta^{\,ab}\,V_{\,b}
=\sum_{b}\delta^{\,ab}\,\frac{\dr\psi}{\dr U^{\,b}},\qquad
\mbox{and}\qquad
\frac{\partial^{2}\cH}{\partial V_{\,a}^{\,2}}
=1,
$$
into \fr{phase-space-compressibility-homogeneous-case}, one has
$$
\kappa_{\,\Psi}^{\,\cH}
=-\,\sum_{a}\frac{1}{\frac{\dr^2\psi}{\dr U^{\,a\,2}}}
\left[\,
\frac{1}{2}\frac{\dr^3\psi}{\dr U^{\,a\,3}}
\,\frac{\dr\psi}{\dr U^{\,a}}
+ \left(\,\frac{\dr^2\psi}{\dr U^{\,a\,2}}\,\right)^{\,2}
\,\right].
$$

If $\psi$ is the sigmoid function, then
$$
\kappa_{\,\Psi}^{\,\cH}
=-\,\sum_{a}\left(
\frac{3}{2}-\frac{\e^{\,U^{\,a}}}{2}
\right)\frac{\e^{\,U^{\,a}}}{(1+\e^{\,U^{\,a}})^{\,2}}.
$$
Notice that 
$$
\left.\kappa_{\,\Psi}^{\,\cH}\right|_{U=0}
=-\frac{n}{4},
$$
with $n=\sum_{a}1$.
%%%%%%%%%%%%%%%%%%
\subsection{Example\,
\ref{example-Hopfield}}
%%%%%%%%%%%%%%%%%%%%%%
The case of Example\,
\ref{example-Hopfield}  
is considered here.

Substituting the differentiation of \fr{example-Cohen-Grossberg-H}, 
$$
\frac{\partial\cH}{\partial V_{\,a}}
=-\,\left[\,
\sum_bJ^{\,ab}V_{\,b}-\frac{U^{\,a}}{R_{\,a}}+I_{\,\ext}^{\,a}\,\right],\quad
\mbox{and}\quad
\frac{\partial^2\cH}{\partial V_{\,a}^{\,2}}
=-\,J^{\,aa}+\frac{1}{R_{\,a}}\frac{\dr U^{\,a}}{\dr V_{\,a}}
=-\,J^{\,aa}+\frac{1}{R_{\,a}}\frac{1}{V_{\,a}(1-V_{\,a})},
$$
into \fr{phase-space-compressibility-homogeneous-case}, one has
$$
\kappa_{\,\Psi}^{\,\cH}
=
-\,\sum_{a}\frac{1}{\frac{\dr^2\psi}{\dr U^{\,a\,2}}}
\left[\,
\frac{1}{2}\frac{\dr^3\psi}{\dr U^{\,a\,3}}
\left(-\,\sum_bJ^{\,ab}V_{\,b}+\frac{U^{\,a}}{R_{\,a}}-I_{\,\ext}^{\,a}\right)
- \left(\,\frac{\dr^2\psi}{\dr U^{\,a\,2}}\,\right)^{\,2}
\left(\,J^{\,aa}-\frac{1}{R_{\,a}}\frac{1}{V_{\,a}(1-V_{\,a})}\,\right)
\,\right],
$$
with $V_{\,a}=\dr\psi/\dr U^{\,a}$.

If $\psi$ is the sigmoid function, then 
$$
\kappa_{\,\Psi}^{\,\cH}
=-\,\sum_{a}
\left[\,
\frac{1-\e^{\,U^{\,a}}}{2(1+\e^{\,U^{\,a}})}
\left(-\,\sum_bJ^{\,ab}V_{\,b}+\frac{U^{\,a}}{R_{\,a}}-I_{\,\ext}^{\,a}\right)
- \,\frac{\e^{\,U^{\,a}}}{(1+\e^{\,U^{\,a}})^2}
\left(\,J^{\,aa}-\frac{1}{R_{\,a}}\frac{1}{V_{\,a}(1-V_{\,a})}\,\right)
\,\right].
$$
This further reduces by the use of 
$$
\frac{\dr^2\psi}{\dr U^{\,a}}
=\frac{\dr V_{\,a}}{\dr U^{\,a}}
=\frac{\e^{\,U^{\,a}}}{(1+\e^{\,U^{\,a}})^2}
=V_{\,a}(1-V_{\,a}),
$$
to 
$$
\kappa_{\,\Psi}^{\,\cH}
=-\,\sum_{a}
\left[\,
\frac{1-\e^{\,U^{\,a}}}{2(1+\e^{\,U^{\,a}})}
\left\{-\,\sum_bJ^{\,ab}V_{\,b}+\frac{U^{\,a}}{R_{\,a}}-I_{\,\ext}^{\,a}\right\}
- \,\frac{\e^{\,U^{\,a}}}{(1+\e^{\,U^{\,a}})^2}
J^{\,aa}+\frac{1}{R_{\,a}}
\,\right],\quad\mbox{with}\quad
V_{\,a}
=\frac{\e^{\,U^{\,a}}}{1+\e^{\,U^{\,a}}}.
$$
For the steady state, $\dot{U}^{\,a}=0$ for all $a$, the term 
$\{\cdots\}$ above 
vanishes. In the steady state, where 
the self-coupling terms vanish, $J^{\,aa}=0$ for all $a$, one has   
$$
\left.\kappa_{\,\Psi}^{\,\cH}\right|_{\mathrm{steady}}
=-\,\sum_{a}\frac{1}{R_{\,a}}.
$$

%%%%%%%%%%%%%%%%%%%%%%%%% 
\section{Conclusions}
\label{section-conclusion}
%%%%%%%%%%%%%%%%%%%%%%%5
This paper has offered a viewpoint  
that a class of  dynamical systems modeling deterministic neural networks 
can be described 
in terms of Hessian and information geometries. 
In this formulation the phase compressibility is shown to be equal to 
the negative of Laplace operator acting on a Lyapunov function. 
Also some explicit forms of them have been shown.

There are some potential future works that follow from this paper.
One is to study a class of stochastic dynamical systems, 
since neural networks are often modeled by 
stochastic models. Then it is of interest to see if the present approach 
can be applied to such stochastic models. 
Another one is to consider a relation between the present formulation on a  
dually flat space and contact Hamiltonian systems in a contact manifold 
\cite{Goto2015,Bravetti2019}. Since the generalized 
Hopfield model 
$(\cM,\Psi,\cH)$ is similar to a Hamiltonian system $(\cS,\omega,H)$
with $\cS$ some even-dimensional manifold, $\omega$ a symplectic $2$-form, 
and $H$ a Hamiltonian function \cite{Silva2008},  
one may explore relations between them. 

We believe that the elucidation of these remaining questions 
together with the present study will develop 
geometric theory of neural network models, and 
that of dynamical systems.  
%%%%%%%%%%%%%%%%%%%%%%%%%%%%%%%%
\section*{Acknowledgments}
%%%%%%%%%%%%%%%%%%%%%%%%%%%%%%%%
The author 
is grateful to 
Hideitsu Hino for 
support for this research. 
Also the author is partially supported 
by JSPS (KAKENHI) grant number 19K03635 and by JST CREST JPMJCR1761.
%%%%%%%%%%%
\appendix
%%%%%%%%
%%%%%%%%%%%%%%%%%%%%%%%%%%% 

%%%%%%%%%%%%%%%%%%%%%
\end{document}